\documentclass[prodmode,acmec]{ec-acmsmall} 

\usepackage{graphicx}
\usepackage{amsmath, amssymb, color}
\usepackage{amsfonts}
\usepackage{amssymb}
\usepackage{mathrsfs}
\usepackage[ruled]{algorithm2e}

\SetAlFnt{\small}
\SetAlCapFnt{\small}
\SetAlCapNameFnt{\small}
\SetAlCapHSkip{0pt}
\IncMargin{-\parindent}

\usepackage[numbers]{natbib}

\begin{document}
\markboth{T. Liu et al.}{Mechanism Design for Cloud Computing}

\title{Online Mechanism Design for Cloud Computing}
\author{Tie-Yan Liu
\affil{Microsoft Research Asia, Beijing }
Weidong Ma
\affil{Microsoft Research Asia, Beijing  }
Tao Qin
\affil{Microsoft Research Asia, Beijing  }
Pingzhong Tang
\affil{Tsinghua University, Beijing }
Bo Zheng
\affil{Tsinghua University, Beijing }}

\begin{abstract}
In this work, we study the problem of online mechanism design for resources allocation and pricing in cloud computing (RAPCC). We show that in general the allocation problems in RAPCC are NP-hard, and therefore we focus on designing dominant-strategy incentive compatible (DSIC) mechanisms with good competitive ratios compared to the offline optimal allocation (with the prior knowledge about the future jobs).
We propose two kinds of DSIC online mechanisms. The first mechanism, which is based on a greedy allocation rule and leverages a priority function for allocation, is very fast and has a tight competitive bound. We discuss several priority functions including exponential and linear priority functions, and show that the former one has a better competitive ratio. The second mechanism, which is based on a dynamic program for allocation, also has a tight competitive ratio and performs better than the first one when the maximum demand of cloud customers is close to the capacity of the cloud provider.
\end{abstract}



\maketitle

\section{Introduction}

Cloud computing is transforming today's IT industry. It offers fast and flexible provisioning of online-accessible computational resources to its customers, thus greatly increases the plasticity and reduces the cost of IT infrastructure. In a cloud computing platform, different kinds of resources are provided to customers, including computing power, storage, bandwidth, database, software, and analytic tools. The statistical multiplexing necessary to achieve elasticity and the illusion of infinite capacity require each of these resources to be virtualized to hide their implementation details. Therefore, a main approach to sell cloud computing resources is through virtual machines (also referred to as {\em instances}): customers can buy and pay for a certain number of virtual machines according to the time of utilization.

A practical problem faced by a cloud service provider is how to appropriately allocate resources and charge customers so as to achieve a balance between profit making and customer satisfaction. Weinhardt et al.,~\citeyear{weinhardtcloud} even claimed that the success of cloud computing in the IT market can be obtained by sorely developing adequate pricing techniques.

\subsection{Existing Pricing Schemes in Cloud Computing}
The most commonly-used pricing scheme in today's cloud computing market is the so-called \emph{pay-as-you-go} model, with which customers pay a fixed price per unit of usage. ~\cite{amazon} utilizes this model and charges a fixed price for each hour of virtual machine usage. Other leading cloud computing products such as ~\cite{azure} and ~\cite{app} also support this pricing model.

\emph{Subscription} is another commonly employed pricing scheme in cloud computing,  with which a customer pays in advance for the services he/she is going to receive for a pre-defined period of time, with a pre-defined fixed price.

Both the \emph{pay-as-you-go} and \emph{subscription} models belong to fixed-price mechanisms with which customers play a passive role. Fixed-price mechanisms are easy to implement. However, they may not be optimal in terms of resource utilization since the demands are dynamically changing. For example, in peak hours, suppose all the resources have been taken by some customers; then even if a new costumer has an emergent task, he/she cannot get the desired resources no matter how much he/she is willing to pay. In  this regard, dynamic and market-based pricing mechanisms are better choices in regulating the supply-demand relationship at market equilibria, and providing satisfactory resource allocation compatible to economic incentives.

As a quick and efficient approach to selling goods at market value, auction-style pricing mechanisms have been widely applied in many fields, reflecting the underlying trends in demand and supply. In fact, an auction-style pricing mechanism  has been adopted by Amazon to dynamically allocate spot instances\footnote{\url{http://aws.amazon.com/ec2/purchasing-options/spot-instances/}} to potential customers. The main advantage of spot instance lies in that it can greatly save the cost of customers because the spot price is usually far below the fixed price. Its disadvantage is the limited applicability, specifically only for those interruption-tolerant tasks: the spot price will go up when more customers come in and a current running task may be interrupted if its bid price is lower than the new spot price. This clearly closes the door to more tasks that are not interruptible, and asks for new kinds of auction-style pricing mechanisms to be invented. This is exactly our focus in this work.

\subsection{Online Mechanism Design for Cloud Computing}

In this paper we study the problem of designing dominant-strategy incentive compatible (DSIC) mechanisms for resource allocation and pricing in cloud computing (RAPCC). In particular, we consider a specific setting of auctions as shown below, which can reflect the nature of cloud computing and distinguish our work from previous studies on auction mechanism design. Please note that designing auction mechanisms in this setting is generally difficult since it combines the challenges of mechanism design (i.e., ensuring incentive compatibility) with the challenges of online algorithms (i.e., dealing with uncertainty about future inputs)~\cite{hajiaghayi2005online}.

\begin{enumerate}
\item A cloud provider has a fixed capacity (denoted as $C\in\mathbb{N}$), i.e., a fixed number of virtual machines (referred to as {\em instances}) in an infinite time interval $T=[0, \infty)$.
\item Customers come and go over time. Each cloud customer has a job to run in the cloud. On behalf of a cloud customer, an agent submits his/her job to the cloud.\footnote{Since customer, job, and agent have one-to-one correspondence in our setting, we will use these terms interchangeably in the following sections of this paper.}
\item An online mechanism is used to determine how to allocate the instances to the agents and how to charge them, without knowledge of future agents who will subsequently arrive.
\item The mechanism is designed to be incentive compatible and to (approximately) maximize the  efficiency (social welfare) of the cloud computing system.
\end{enumerate}

To be more specific, we explain the details of the above setting as follows.

We use $J$ to represent the set of jobs.  Let $r_i$ be the release time of job $i$ and $d_i$ be the deadline of the job. The private information (i.e., {\em type}) of agent $i$ is characterized by a tuple $\omega_i=( n_i, l_i, v_i)\in\Omega$, where $n_i$ is the number of instances required by the job, $l_i$ is the length of time required by running the job, and $v_i$ is the value that the agent can get if the job is completed. Here we say a job $i$ is {\em complete} if it is allocated with at least $n_i$ instances for $l_i$ units of time continuously between its release time $r_i$ and deadline $d_i$. The space $\Omega$ consists of all possible agent types. Note that types are private information: agent $i$ observes its type only at time $r_i$, and nothing is known about the job before $r_i$.

Note that in our setting, interruption of jobs is allowed but the interrupted job gets zero value. Once a job is interrupted, the resources already spent on the job are wasted: when the job is restarted, another $l_i$ units of time will be needed for its completion. Partial allocation is not legitimate, i.e., allocating $x_i$ instances to agent $i$ is useless if $x_i<n_i$. The deadline of each job is hard, which means that no value is obtained for a job that is completed after its deadline.

Similar to \cite{ng2003virtual}, we do not consider agents' temporal strategies and assume that agents will not misreport $r_i$ and $d_i$ in this work. We leave agents' temporal strategies to future work. For convenience, we usually combine $r_i$, $d_i$ with $\omega_i$ and denote them by $\theta_i$, i.e., $\theta_i=(r_i, d_i, n_i, l_i, v_i)$.

In practice, a cloud platform usually specifies the shortest and longest lengths of a job, and directly rejects those jobs whose lengths are out of this range. To reflect this and without loss of generality, we assume the minimum and maximum length of a job to be $1$ and $\kappa$ respectively. For the convenience of analysis, we assume $\kappa$ is an integer.

We study direct revelation mechanisms ~\cite{myerson1981optimal}, in which each agent participates by simply reporting his/her type. Please note that agents are selfish and rational. Therefore, they may misreport their type in order to be better off.  We use $\hat{\theta_i}$$=(r_i, d_i, \hat{n}_i, \hat{l}_i, \hat{v}_i)$ to represent agent $i$'s report. It is easy to see that the misreport of a shorter length is a dominated strategy; otherwise, his/her job cannot be completed when the provider allocates $\hat{n}_i$ instances with time length $\hat{l}_i<l_i$. Therefore, we assume that agents will not misreport a shorter length. Based on the reports of the agents, the mechanism determines how to allocate and price the computing resources.

Let $x$ be an allocation function and $x_i(t)$ be the number of instances allocated to job $i\in I(t)$ at time $t$, where $I(t)$ is the set of jobs available to the mechanism at time $t$. We say $x$ is feasible if $\sum_{i\in I(t)}x_i(t)\le C, \forall t$. For a job $i$ and an allocation function $x$, let $q_i(x) = 1$ if it is completed, otherwise $q_i(x) = 0$. The value of agent $i$ extracted from allocation $x$ is $q_i(x)v_i$. The efficiency (social welfare) of the allocation function $x$ is $W(x)= \sum_iq_i(x)v_i$.

Let $p$ be a payment rule and $p_i$ be the amount of money agent $i$ needs to pay to the cloud service provider. We assume that agents have quasi-linear utilities, i.e., the utility of
agent $i$ for the allocation function $x$ and the payment rule $p$ is $u_i(x,p)=q_i(x)v_i-p_i$.

A mechanism $M=(x,p)$ is said to be {\em dominant-strategy incentive compatible (DSIC)} if, for any agent $i$, regardless of the behaviors of other agents, truthfully reporting his/her own type can maximize his/her utility. The mechanism is said to be {\em individual rational} if for each job $i$, $u_i(x,p)\geq 0$.

Hajiaghayi et al.~\citeyear{hajiaghayi2005online} provide a simple characterization for DSIC mechanisms by monotonicity, which is rephrased as Lemma \ref{lem:mono}.

\begin{definition}\label{mono}
We say that a type $\theta_i=(r_i, d_i, n_i, l_i, v_i)$ {\em dominates} the type $\theta'_i=(r'_i, d'_i,  n'_i, l'_i, v'_i)$, denoted $\theta_i\succ\theta'_i$
if $r_i\leq r'_i$, $d_i\geq d'_i$, $n_i\leq n'_i$, $l_i\leq l'_i$, and $v_i>v'_i$.
An allocation function $x$ is {\em monotone} if for every agent $i$, we have $q_i(x_i(\theta_i,\theta_{-i}))\geq q_i(x_i(\theta'_i,\theta_{-i})), \forall \theta_i\succ\theta'_i , \forall\theta_{-i}$.
\end{definition}

\begin{lemma}\label{lem:mono}\cite{hajiaghayi2005online}
For any allocation function $x$, there exists a payment rule $p$ such that the mechanism $(x, p)$ is DSIC if and only if $x$ is monotone.
\end{lemma}

We are interested in designing DSIC and individual rational mechanisms. In addition, we also would like the mechanism to have good performance in (approximately) maximizing the social welfare of the cloud computing system.

In particular, we use the concept of \emph{competitive ratio} \cite{lavi2000competitive} to evaluate the performance of a mechanism (see Definition \ref{def1}), which compares the social welfare implemented by the mechanism (without any knowledge of future jobs) with that of the optimal offline allocation (with the prior knowledge of future jobs).

\begin{definition}\label{def1}
An online mechanism $M$ is (strictly)
$c$-{\em competitive} if there does not exist an job sequence $\theta$ such that
$c\cdot W(M, \theta)<OPT(\theta)$, where $OPT(\theta)$ denotes the social welfare of the optimal offline allocation. Sometimes we also say that $M$ has a competitive ratio of $c$.
\end{definition}

\subsection{Our Results}
The main results of our work are summarized as follows.
\begin{enumerate}
\item (Section \ref{hlb}) We show that the allocation problem in our setting is NP-hard through a reduction from the Knapsack problem to our problem.

\item (Section~\ref{gbm}) We design a DSIC mechanism $\Gamma_G$ based on the greedy algorithm proposed for the Knapsack problem~\cite{vazirani2001approximation}. In $\Gamma_G$, we assign a priority score to each active job and then allocate resources based on the virtual values of the active jobs computed from priority scores. We study several different priority functions and obtain the following results.
\begin{enumerate}
\item When assigning priorities according to an exponential function, the competitive ratio of $\Gamma_G$ is tightly bounded by $\frac{h}{h-1}\cdot \frac{\chi}{1-\chi^{-1/\kappa}}+1$ when $h\geq 2$, where $h$ is the rounded ratio between the capacity $C$ and the maximum number of instances demanded  by a customer,  and $\chi>1$ is the base of the exponential function. Specifically, when we choose $\chi=(\frac{\kappa+1}{\kappa})^\kappa$, $\Gamma_G$ achieves the best competitive ratio of $\frac{h}{h-1}\cdot (\kappa+1)(1+\frac{1}{\kappa})^{\kappa}+1$.  And for the special case with the capacity $C= 1$ (which is identical to the conventional online real-time scheduling problem), the competitive ratio of $\Gamma_G$ is tightly bounded by $\frac{\chi}{1-\chi^{-1/\kappa}}+1$.
\item When assigning priorities according to a linear function, the competitive ratio of $\Gamma_G$ is lower bounded by $\frac{h}{h-1}\cdot (\sqrt{2\kappa(\kappa+1)}+\frac{3}{2}\kappa+\frac{1}{2})+1$. This result implies that the exponential priority is better than the linear priority, since $\sqrt{2\kappa(\kappa+1)}+\frac{3}{2}\kappa+\frac{1}{2}$ is greater than $(\kappa+1)(1+\frac{1}{\kappa})^{\kappa}$.
\item When assigning priorities according to a general non-decreasing function $f(\delta)$, the competitive ratio of $\Gamma_G$ is lower bounded by $\frac{h}{h-1}\cdot(\sqrt{\kappa}+1)^2+1$, when $f(\delta)$  satisfies
$f(0)=1$, where $\delta$ is the completed fraction of a job.
\end{enumerate}

\item (Section \ref{dpm}) We design another DSIC mechanism $\Gamma_D$ based on the dynamic program proposed for the Knapsack problem \cite{martello1990knapsack}. This mechanism has
a competitive ratio of exactly $n_{\max}\cdot \frac{\chi}{1-\chi^{-1/\kappa}}+1$, where $n_{\max}$ is the maximum number of instances demanded by a customer. Comparing the $\Gamma_G$, this mechanism has a much better competitive ratio when $n_{\max}$ is close to the capacity $C$.

\end{enumerate}

\subsection{Related Work}
\cite{lavi2000competitive} coined the term ``online auction'' and initiated the study
of incentive compatible mechanisms in dynamic environments with the computer science
literature. \cite{friedman2003pricing} initiated the study of VCG-based online
mechanisms and coined the term ``online mechanism design''. Later on, MDP-based approaches \cite{parkes2003mdp, parkes2004approximately} have been applied to study the online VCG mechanism, in which prior knowledge is assumed about the future arrivals. Different from those works, the online setting concerned here is model-free (i.e., no knowledge is assumed about future), since cloud computing is a very dynamic environment and it is difficult to predict future jobs, especially in the auction-style setting.

Model-free online setting has been studied in \cite{porter2004mechanism, hajiaghayi2005online}, which design DSIC mechanisms for online scheduling of a single, re-usable resource. A competitive ratio of $(\sqrt{k}+1)^2+1$ is achieved  with respect to the optimal efficiency in \cite{porter2004mechanism}, where $k$ is the ratio of maximum to minimum value density (value divided by processing time) of a job, and the ratio is proved to be optimal for deterministic mechanisms. \cite{ hajiaghayi2005online} provided a randomized mechanism whose efficiency competitive ratio is $O(log(\kappa))$ (recall that $\kappa$ is the ratio of the maximum job length to the minimum job length. Unlike these works, in our problem, we have multiple instances to sell at each time step and each job demands multiple instances and multi-unit time.

Recently, \cite{gerding2011online, robu2012online, robu2013online} studied online mechanisms for electric vehicle charging. In this problem, agents are assumed to be with multi-unit demand and non-increasing marginal values. That is, the first unit allocated to an agent have a higher (or equal) marginal value for the agent compared to any subsequent units. The difference between those works and our model lies in the definition of agents' utilities: in our problem an agent can get value if and only if his/her job is fully completed, while in their problems, an agent can collect value even if his/her demand is only partially fulfilled.  Note that our setting is closer to cloud computing, in which agents want their jobs fully completed. Therefore, the mechanisms designed in those works do not work for our problem.

Another related work is \cite{ng2003virtual}, which studied  the problem of designing fast and incentive-compatible exchanges for dynamic resource allocation problems in distributed
systems. Different from our work, they considered a two-sided market, in which there are both request agents (consuming resources) and service agents (providing resources) coming sequentially. Their setting on the side of request agents is very similar to us: they ignored the temporal strategies and considered a three dimensional type (i.e., size of the job, length of the time, and the value of the job) for request agents. Because their model is more complex and needs to consider both the strategies of buyer side and seller side, they focused on designing incentive compatible mechanisms without theoretical analysis on the efficiency of the mechanism. In contrast, we design DSIC mechanisms for our problem and (almost) tight competitive bounds are obtained.

\section{Computational Complexity}\label{hlb}
Before presenting our mechanisms, we first consider the computational complexity of the allocation
problem in our model.
\begin{theorem}\label{the:np}
The allocation problem in our model is NP-hard. More precisely, the decision problem of whether the optimal allocation
has social welfare of at least $k$ (where $k$ is an additional part of the input) is
NP-complete.
\end{theorem}
\begin{proof}
We show that any knapsack problem can be reduced to the allocation problem in our model.

Consider a knapsack with size $C\in\mathbb{Z}^+$ and a set of items denoted as $S=\{1,\ldots, n\}$. Each item $i$ in the set has size $s_i\in \mathbb{Z}^+$ and profit $v_i\in\mathbb{R}^+$. The knapsack problem is whether one can pack a subset of items into the knapsack with total profit greater than $k$.

Given such an instance of knapsack problem, we will build a cloud resource allocation problem from it as follows:
A cloud provider has $C$ virtual machines. There are $n$ agents, and agent $i$'s type is $\theta_i=\{0, 1, s_i,1, v_i\}$. Now notice that a yes/no answer to the decision problem of the cloud resource allocation corresponds to a yes/no answer to the decision problem of knapsack problem, and vice versa. Since the knapsack problem is NP-complete, this concludes the NP-hardness of the allocation problem in cloud computing.
\end{proof}

\section{A Greedy Mechanism}\label{gbm}
In this section, we design a greedy mechanism for resource allocation and pricing in cloud computing (RAPCC) and prove its competitive efficiency.

For any time $t$, we use $\delta _i\leq 1$ to denote the fraction that job $i$ has been continuously processed before time $t$ (i.e. he/she has received an allocation at time $t - \delta_i l_i$ and has not been interrupted after that), and we call $\delta _i$ the rate of completeness.
We say a job $i$ is \emph{feasible} at time $t$ if
\begin{enumerate}
\item it has been released before $t$, i.e., $r_i\le t$;
\item it has enough time to be completed before its deadline, i.e., $d_i-t\geq (1-\delta_i)l_i$; and
\item it has not been completed yet, i.e.,  $\delta_i<1$.
\end{enumerate}
We use $J_F(t)$ to denote the set of feasible jobs at $t$.

The basic idea of the proposed mechanism is that we assign a priority score to each feasible job, compute a virtual value for each feasible job, allocate the resources to the feasible jobs according to their virtual values at each critical time point, and charge each agent at his/her deadline according to his/her critical value \cite{porter2004mechanism} if his/her job is completed. We say $t$ is a \emph{critical time point} if some new job arrives at time $t$ or some existing job is completed at time $t$. Given an allocation function, the \emph{critical value}  of a job is the minimum reported value that ensures it can be completed by its deadline. Note that we do not charge an agent if his/her job is not completed before his/her deadline.

\begin{algorithm}[H]
\For{all critical time point $t$ in the ascending order}{
$J_F\leftarrow J_F(t)$\;
$\forall i \in J_F$, update its virtue value density $\rho'_i=\frac{v_i}{n_i}f(\delta_i)$ and virtue value $v'_i=v_if(\delta_i)$\;
Re-number jobs in $J_F$ by the descending order of $\rho'_i$\;
\eIf{there exists $k$ such that the size of the first $k$ jobs exceeds $C$}
{
\eIf{$\sum_{i=1}^{k-1}v'_i\geq v'_k$}{
Run the job set$\{1,\ldots, {k-1}\}$ \;}{Run job $k$\;}}
{Run the job set $J_F$\;}
}
\caption{The greedy allocation rule of $\Gamma_G$}
\end{algorithm}

The allocation rule\footnote{Since the payment rule is very simple, we omit it.} of mechanism $\Gamma_G$ is shown in Algorithm 1, in which $f()$ is a non-decreasing priority function satisfying $f(0)=1$.

There can be different ways to assign priority scores to jobs. We study three priority functions in the following subsections.

\subsection{Exponential Priority Functions}
In this subsection, we study the mechanism $\Gamma_{G}$ with an exponential priority function:
  $$f(\delta) = \chi^{\delta},$$
where $\chi>1$ is an input parameter.

It is easy to see that with such a priority function, the allocation rule is monotone. According to Lemma \ref{lem:mono}, the mechanism $\Gamma_{G}$ is dominant-strategy incentive compatible.

Next, we prove a tight competitive ratio for the mechanism.

\begin{theorem}\label{main1}
Assume $C\geq h\cdot n_{\max}$, where $h\geq 2$ is an integer. The competitive ratio of $\Gamma_G$ with an exponential priority function is $\frac{h}{h-1}\cdot \frac{\chi}{1-\chi^{-1/\kappa}}+1$.
\end{theorem}

We prove the theorem with two lemmas. First, we use an example to prove the competitive ratio of $\Gamma_{G}$ is at least $\frac{h}{h-1}\cdot \frac{\chi}{1-\chi^{-1/\kappa}}+1$. Then, we prove the competitive ratio is upper bounded by $\frac{h}{h-1}\cdot \frac{\chi}{1-\chi^{-1/\kappa}}+1$.

\begin{lemma}\label{lem1}
The competitive ratio of $\Gamma_{G}$ with an exponential priority function is at least $\frac{h}{h-1}\cdot \frac{\chi}{1-\chi^{-1/\kappa}}+1$.
\end{lemma}
\begin{proof}
We prove this lemma by an example. For the convenience of analysis, we assume $\kappa$ is an integer. Consider an example with $C=h\cdot n_{\max}$ and two types of jobs: long and
short. The length of long jobs is $\kappa$, while the length of short jobs is $1$. The jobs are released by groups. Let $p$ be a large integer, and we have $p+1$ groups of long jobs
and $p\kappa$ groups of short jobs, respectively.

The first group of long jobs (denoted as $J_0^l$) consists of $h$ long jobs with type $\theta_0^l=(0, \kappa, n_{\max}, \kappa, n_{\max})$.

The $(i+1)$-th group of long jobs (denoted as $J_i^l$) consists of $h$ long jobs with type $\theta_i^l=(i(\kappa-\epsilon), (i+1)\kappa, n_{\max}, \kappa, n_{\max}\cdot\chi^i)$, where
$p-2 \geq i \geq 1$.

The $p$-th group of long jobs (denoted as $J_{p-1}^l$) consists of $h-1$ long jobs with types $\theta_{p-1}^{l_1}=((p-1)(\kappa-\epsilon), (p+2)\kappa, n_{\max}, \kappa, n_{\max}\cdot
\chi^{p-1})$, and one long job with type $\theta_{p-1}^{l_2}=((p-1)(\kappa-\epsilon), (p+2)\kappa, 1, \kappa, \chi^{p-1})$.

The $(p+1)$-th group of long jobs (denoted as $J_p^l$) consists of $h$ long jobs with type $\theta_p^l=(p(\kappa-\epsilon), (p+1)\kappa, n_{\max}, \kappa,
n_{\max}\cdot\chi^{p-\epsilon}-\delta)$.

Here $\epsilon$ and $\delta$ are small constants satisfying $p\epsilon \ll 1$ and $\delta \ll \epsilon$.

In the meanwhile, we have $p\kappa$ groups of short jobs as follows.

The $(j+1)$-th group of short jobs (denoted as $J_s^j$) consists of $h$ short jobs with types $\theta_s^j=(j, j+1, n_{\max}, 1, n_{\max}\cdot(\chi^{j/\kappa}-\delta/\kappa))$. Here
$j=0,1,\ldots,p\kappa-1$.

It can be verified that only the jobs in group $J_{p-1}^l$ can be completed in the mechanism, with a
social welfare $\sim ((h-1)\cdot n_{\max}+1)\chi^{p-1}$. While in the optimal allocation, all the short jobs will be completed, and after that, group $J_p^l$ and group
$J_{p-1}^l$ will be completed successively, with a social welfare  $\sim h\cdot n_{\max}\sum_{j=0}^{p\kappa-1}\chi^{j/k}+h\cdot n_{\max}\cdot\chi^{p}+ ((h-1)\cdot n_{\max}+1)\chi^{p-1}
=h\cdot n_{\max}\cdot\frac{1-\chi^{p+1/\kappa}}{1-\chi^{1/\kappa}}+((h-1)\cdot n_{\max}+1)\chi^{p-1}$.
 Therefore, the competitive ratio of our mechanism is at
least $\frac{h\cdot n_{\max}}{(h-1)\cdot n_{\max}+1}\cdot \frac{1-\chi^{p+1/\kappa}}{(1-\chi^{1/\kappa})\chi^{p-1}}+1=\frac{h\cdot n_{\max}}{(h-1)\cdot n_{\max}+1}\cdot \frac{\chi-\chi^{-1/\kappa-p+1}}{1-\chi^{-1/\kappa}}+1$, which tends to $\frac{h}{h-1}\cdot
\frac{\chi}{1-\chi^{-1/\kappa}}+1$, when $p \rightarrow \infty$ and $n_{\max} \rightarrow \infty$.
\end{proof}

\begin{lemma}\label{lem2}
The competitive ratio of $\Gamma_{G}$ with an exponential priority function is at most $\frac{h}{h-1}\cdot \frac{\chi}{1-\chi^{-1/\kappa}}+1$.
\end{lemma}
\begin{proof}
Similar to~\cite{hajiaghayi2005online},
we will charge the values of winning jobs in an optimal allocation
(denoted as OPT) to winning jobs in our mechanism. Here a winning job in an allocation means the job is completed in the allocation. We assume, without loss
of generality, that OPT does not interrupt any job.

We draw a line $\ell$ which represents a capacity of $\frac{h-1}{h}C$ instances (Fig.~\ref{fig:3}).
 \begin{figure}[htpb]
\centerline{  \includegraphics[scale=0.25]{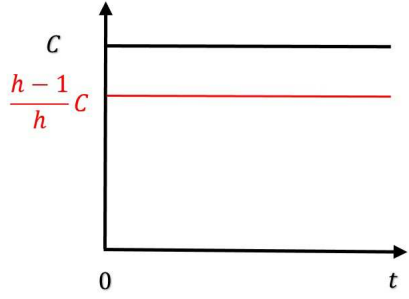}}
\caption{The line $\ell$ which represents the capacity of $\frac{h-1}{h}C$}
\label{fig:3}       
\end{figure}
For any winning agent $i$ in OPT, if he/she is also a winner in our mechanism,
then his/her value is charged to himself/herself.
Otherwise, consider the time $t$ at which $i$ is allocated the instances in OPT. At this time,
our mechanism allocates at least $\frac{h-1}{h}C$ instances to agents, since $n_{\max}\leq \frac{C}{h}$ and
$i$ is not allocated. We sort the jobs (denote $J_i$) that are allocated at
time $t$ by decreasing order of $\frac{v_i}{ {n}_i}\cdot \chi^{\delta_i}$. We let job $1$ (after sorting) get the bottom $n_1$ instances,
job $2$ is allocated above job $1$, and so forth. Using $J_{it}\subseteq J_i$ to denote the set of jobs
that are allocated under the line $\ell$ (if the line cuts some job $j\in J_i$, we only consider the part that under $\ell$
and use $j'$ to represent this part).
We first temporarily charge the value of $i$ to all the jobs in $J_{it}$, then each job $j\in J_{it}$ is temporarily charged $\frac{n_jh}{(h-1)C}v_i$ ($\leq \frac{n_ih}{(h-1)C}v'_j$).
A job $j\in J_{it}$ might be interrupted in our algorithm. If he/she is not interrupted, then he/she is finally charged $\frac{n_jh}{(h-1)C}v_i$.
If he/she is interrupted at time $t'$, then some jobs that were under $\ell$ before $t'$ may
be allocated above $\ell$ at $t'$. We use $J_{int}$ to denote all these jobs and jobs that are interrupted at $t'$.
We pass all the temporary charge of $j\in J_{int}$ to jobs that are newly allocated under line $\ell$ at $t'$, and
other jobs in $J_{it}\backslash J_{int}$ keep their temporary charge.

Note that after the interruption the total value
of jobs that under $\ell$ will not decrease, since $n_{\max}\leq \frac{C}{h}$ and jobs are sorted by decreasing of $\rho'$. Therefore,
after the interruption, each job $j$ under $\ell$ has a temporary charge of at most $\frac{n_jh}{(h-1)C}v_i\leq \frac{n_ih}{(h-1)C}v'_j$.
We continue this chain until all the temporary charge are finally charged.
We now calculate the maximum total value charged to agent $j$ with
value $v_j$ who wins at time $t$ in our mechanism.

If job $j$ is completed in $OPT$, there is a charge of $v_j$. Divide all jobs in $OPT$ whose value
is charged to $j$ to different groups according to their start time in $OPT$ by the following rule:

Consider a job $i$ in $OPT$ whose value is charged to $j$. Let $t' = t- \sigma_i$ be the time at which job $i$ receives an allocation in $OPT$, then we say $i$ is in group $\sigma_i$.
It is clear from the mechanism that $\sigma_i > -l_j$.

When $\sigma_i\leq 0$, it is easy to see that, the value of job $i$ is at most $\chi^{-\sigma_i/l_j}\frac{n_i v_j}{n_j}$. Thus, the total charge from group $\sigma_i$ is at most $\sum_{i}n_i\frac{v_j}{n_j}\frac{n_jh}{(h-1)C}\chi^{-\sigma_i/l_j}\leq \frac{h}{h-1}\chi^{-\sigma_i/l_j}v_j$.

When $\sigma_i>0$, we now claim that the value of $i$ is at most $\chi^{-\sigma_i/k}\frac{n_i v_j}{n_j}$, and the reason is as follows: When $\sigma_i>0$, job $i$ is released before $j$. There exists two scenarios which make job $i$ interrupted.
\begin{enumerate}
\item
Assume that job $i$ is interrupted by job $i_2$ after being allocated for $\sigma^1_i$ units of time, and then job $i_2$ is interrupted by job $i_3$ after $\sigma^2_i$ units of time, and so on. The last job in this chain is $i_\tau$ which is interrupted by job $j$ after $\sigma^\tau_i$ units of time. Then we know from our mechanism that
$\sigma^1_i+\sigma^2_i+\ldots+\sigma^\tau_i=\sigma_i$ and $\frac{v_i}{n_i} \chi^{\sigma^1_i /l_{i}}\leq \frac{v_{i_2}}{n_{i_2}}$,
$\frac{v_{i_2}}{n_{i_2}}\chi^{\sigma^2_i /l_{i_2}}\leq \frac{v_{i_3}}{n_{i_3}}, \ldots$,
and $\frac{v_{i_\tau}}{n_{i_\tau}} \chi^{\sigma^\tau_i/l_{i_\tau}}\leq \frac{v_j}{n_j}$, which combining with $l_{\max}=\kappa$ implies that
$v_i\leq \chi^{-\sigma_i/\kappa}\frac{n_i v_j}{n_j}$. Thus, the total charge from group $\sigma_i$ is at most $\sum_{i}n_i\frac{v_j}{n_j}\frac{n_jh}{(h-1)C}\chi^{-\sigma_i/\kappa}\leq \frac{h}{h-1}\chi^{-\sigma_i/\kappa}v_j$

\item
Assume that before job $i$ is released, job $i_2$ has been processed for $z$ units of time, then job $i_2$ is interrupted by job $i_3$ after $\sigma^2_i$ units of time, and so on.  The last job in this chain is $i_\tau$ which is interrupted by job $j$ after $\sigma^\tau_i$ units of time. Then we know from our mechanism that
$\sigma^2_i+\sigma^3_i+\ldots+\sigma^\tau_i=\sigma_i$ and $\frac{v_i}{n_i} \chi^{-z/l_{i_2}}\leq \frac{v_{i_2}}{n_{i_2}}$,
$\frac{v_{i_2}}{n_{i_2}}\chi^{(z+\sigma^2_i)/l_{i_2}}\leq \frac{v_{i_3}}{n_{i_3}}, \ldots$,
and $\frac{v_{i_\tau}}{n_{i_\tau}} \chi^{\sigma^\tau_i/l_{i_\tau}}\leq \frac{v_j}{n_j}$, which combining with $l_{\max}=\kappa$ implies that
$v_i\leq \chi^{-\sigma_i/\kappa}\frac{n_i v_j}{n_j}$. Thus, the total charge from group $\sigma_i$ is at most $\sum_{i}n_i\frac{v_j}{n_j}\frac{n_jh}{(h-1)C}\chi^{-\sigma_i/\kappa}\leq \frac{h}{h-1}\chi^{-\sigma_i/\kappa}v_j$
\end{enumerate}

Also, the value of $\sigma_i$ for any two such groups must be apart by at least $l_{min}=1$, so $\sigma_i=-l_j+i\cdot 1$, for $i=0,1,2,\ldots,+\infty$. Therefore, the total charge to $j$ is at most
\begin{align}
\nonumber
&v_j+\frac{h}{h-1}\cdot \sum_{i: \sigma_i\geq0}^\infty \chi^{-\frac{\sigma_i}{\kappa}} v_j + \frac{h}{h-1}\cdot \sum_{i: \sigma_i<0} \chi^{-\frac{\sigma_i}{l_j}} v_j\\
\nonumber
\leq& v_j+\frac{h}{h-1}\cdot \sum_{i: \sigma_i=0}^\infty \chi^{-\frac{\sigma_i}{\kappa}} v_j + \frac{h}{h-1}\cdot \sum_{i:\sigma_i=-\kappa}^{i:\sigma_i=-1} \chi^{-\frac{\sigma_i}{\kappa}} v_j\\
\nonumber
=&(\frac{h}{h-1}\cdot \frac{\chi}{1-\chi^{-\frac{1}{\kappa}}}+1)v_j
\end{align}

This shows that our algorithm is $(\frac{h}{h-1}\cdot \frac{\chi}{1-\chi^{-\frac{1}{\kappa}}}+1)$-competitive.
\end{proof}

Combining Lemma~\ref{lem1} and \ref{lem2} completes the proof of Theorem~\ref{main1}.

The following lemma shows that if $n_{\max}\rightarrow C$, $\Gamma_{G}$ does not have a constant competitive ratio.

\begin{lemma}\label{lem3}
The competitive ratio of $\Gamma_{G}$  with an exponential priority function can be arbitrarily bad when $n_{\max}=C$.
\end{lemma}
\begin{proof}
Consider the following example: Let $p$ be a large integer, $\epsilon$ and $\mu$ be small constants satisfying $p\epsilon\ll1$ and $\mu\ll\epsilon$.
There are $4p$ jobs. Each job has zero laxity and unit length, and their types are as follows. $\theta_1=\theta_2=(0, 1, \frac{C}{2}, 1, 2^p)$, $\theta_3=(\mu, \mu+1, \frac{C}{2}+1, 1, 2^p+\frac{2^{p+1}}{C}+\epsilon)$, and
$\theta_4=(2\mu, 1+2\mu, C, 1, 2^p+\frac{2^{p+1}}{C}+2\epsilon)$. (we choose the value density of job $3$ ``just'' larger than job $1$ and $2$'s, and job $3$ preempts job $1$ and $2$. Then we choose the value of job $4$ ``just'' larger than job $3$'s, and job $4$ preempts job $3$.)
Generally, for $i=4j+1$ ($p-1\geq j\geq 1$), $\theta_i=\theta_{i+1}=(3j\mu, 3j\mu+1, \frac{C}{2}, 1, \frac{v_{i-1}}{2}+\epsilon)$, $\theta_{i+2}=((3j+1)\mu, (3j+1)\mu+1, \frac{C}{2}+1, 1, \frac{v_{i-1}}{2}+\frac{v_{i-1}}{C}+2\epsilon)$, and $\theta_{i+3}=((3j+2)\mu, (3j+2)\mu+1, C, 1, \frac{v_{i-1}}{2}+\frac{v_{i-1}}{C}+3\epsilon)$. We can verify that in $\Gamma_G$ only
the last job (i.e., job $4p$) wins whose value is about $2^{p+1} (\frac{1}{2}+\frac{1}{C})^p$. While in OPT, the first two jobs win, and the social welfare is $2^{p+1}$. Therefore the competitive ratio
of $\Gamma_G$  with an exponential priority function is at least $(\frac{1}{2}+\frac{1}{C})^p$. As $C$ is usually a large number in practical cloud computing, we can assume $C>2$ and have $(\frac{1}{2}+\frac{1}{C})$ smaller than $1$. Therefore, the competitive ratio tends to infinity when $p$ tends to infinity.
\end{proof}

\begin{remark}
Fortunately, in the practice of cloud computing, the demand $n_i$ of an individual cloud customer is usually much smaller than the capacity of the cloud, and therefore the proposed mechanism is expected to perform well in real-world cloud computing market.
\end{remark}

Now we consider a special case of our model, where $C=n_{\max}=1$. We have the following theorem for this special case.
\begin{theorem}\label{main2}
Assume $C=n_{\max}=1$, $\Gamma_{G}$  obtains a tight competitive ratio $\frac{\chi}{1-\chi^{-1/\kappa}}+1$.
\end{theorem}

\begin{proof}
We prove the theorem by charging the value of any completed job in an optimal allocation (denoted as $OPT$) to a completed job in our mechanism. For any completed job $i$ in $OPT$, if it is also completed in our mechanism, then its value is charged to itself. Otherwise, consider the time $t$ at which $i$ is completed in $OPT$. At this time, our mechanism has been processing another job $j_0$. This job might be preempted in our mechanism. If it is preempted, let $j_1$ be the job that preempts it. We continue this chain until we reach a job $j_k$ which is not preempted, and charge the value of job $i$ to this job.

We now calculate the maximum total value charged to a job $j$ with value $v_j$, which is released at time $t$ and will be completed in our mechanism. If job $j$ is completed in $OPT$, there is a charge of $v_j$. Consider a job $i$ in $OPT$ whose value is charged to $j$. Let $t' = t- \sigma_i$ be the time at which $i$ is processed in $OPT$.  Similar to proof of Theorem~\ref{main1}, we easily know $\sigma_i > -l_j$, and what's more, when $\sigma_i\leq0$, the value of $i$ is at most $\chi^{-\sigma_i/l_j}v_j$, and when $\sigma_i>0$, the value of $i$ is at most $\chi^{-\sigma_i/k}v_j$,

Also, the value of $\sigma_i$ for any two such $i$'s must be apart by at least $l_{min}=1$, so $\sigma_i=-l_j+i\cdot 1$, for $i=0,1,2,\ldots,+\infty$. Therefore, the total charge to $j$ is at most
\begin{align}
\nonumber
&v_j+\sum_{i: \sigma_i\geq0}^\infty \chi^{-\frac{\sigma_i}{\kappa}} v_j + \frac{h}{h-1}\cdot \sum_{i: \sigma_i<0} \chi^{-\frac{\sigma_i}{l_j}} v_j\\
\nonumber
\leq& v_j+\sum_{i: \sigma_i=0}^\infty \chi^{-\frac{\sigma_i}{\kappa}} v_j + \frac{h}{h-1}\cdot \sum_{i:\sigma_i=-\kappa}^{i:\sigma_i=-1} \chi^{-\frac{\sigma_i}{\kappa}} v_j\\
\nonumber
=&(\frac{\chi}{1-\chi^{-\frac{1}{\kappa}}}+1)v_j
\end{align}
This shows that our mechanism is ($\frac{\chi}{1-\chi^{-1/\kappa}}+1$)-competitive.

We give an example below to show that the above analysis is tight. For the convenience of analysis, we assume $\kappa$ is an integer and $l_{\min}=1$. In our example, there are two types of jobs: long and short. The length of long jobs is $\kappa$, while the length of short jobs is $1$. Let $p$ be a large integer, the number of long and short jobs are $p+1$ and $p\kappa$, respectively. The fist long job $J_0^l$ is released at time $0$, and its type is
$\theta_0^l=(0, \kappa, 1, \kappa, 1)$. For $p-2 \geq i \geq 1$, job $J_i^l$ has type $\theta_i^l=(i(\kappa-\epsilon), (i+1)\kappa, 1, \kappa, \chi^i)$. Job $J_p^l$ has type $\theta_p^l=(p(\kappa-\epsilon), (p+1)\kappa, 1, \kappa, \chi^{p-\epsilon}-\delta)$. Here, $\epsilon$ and $\delta$ are small constants satisfying $p\epsilon \ll 1$ and $\delta \ll \epsilon$. We also have long job $J_{p-1}^l$, whose type is $\theta_{p-1}^l=((p-1)(\kappa-\epsilon), (p+2)\kappa, 1, \kappa, \chi^{p-1}+\delta)$. In the meanwhile, we have short jobs as follows. For $j=0,1,\ldots,p\kappa-1$, we denote $J_j^s$ as the $(j+1)$th short job, whose type is $\theta_j^s=(j, j+1, 1, 1, \chi^{j/\kappa}-\delta/\kappa)$. for $j=0,1,\ldots,p\kappa-1$.

It can be verified that only one job $J_{p-1}^l$ can be completed in our mechanism, with a social welfare $\sim \chi^{p-1}$. While in optimal solution, all the short jobs will be completed, and after that, $J_p^l$ and $J_{p+1}^{l}$ will be completed successively, with a social welfare
$\sim \frac{1-\chi^{-1}}{1-\chi^{-1/\kappa}}(1+\chi+,\ldots,+\chi^{p-1}+\chi^p)+\chi^{p-1}$. Therefore, the competitive ratio of our mechanism is at least $\frac{1-\chi^{-1}}{1-\chi^{-1/\kappa}}(\chi^{-(p-1)}+\chi^{-(p-2)},\ldots,+1+\chi)+1$, which tends to $\frac{\chi}{1-\chi^{-1/\kappa}}+1$, when $p \rightarrow \infty$.
\end{proof}

In many situations, we know the maximum length $\kappa$ of any job in advance (e.g., specified by the cloud provider), then our mechanism can choose a best $\chi$ to obtain the best competitive ratio. we have the following proposition.

\begin{proposition}\label{exp}
If $\kappa$ is known, when $\chi=(\frac{\kappa+1}{\kappa})^\kappa$, the mechanism $\Gamma_{G}$  obtains the best competitive ratio. When $C\geq h\cdot n_{\max}$, where $h\geq 2$ is an integer, this competitive ratio is $\frac{h}{h-1}\cdot (\kappa+1)(1+\frac{1}{\kappa})^{\kappa}+1$. When $C=n_{\max}=1$, this competitive ratio is $(\kappa+1)(1+\frac{1}{\kappa})^{\kappa}+1$.
\end{proposition}

\subsection{General Non-decreasing Priority Functions}
We consider any general non-decreasing priority function, and show that the competitive ratio of $\Gamma_G$ is greater than $(\sqrt{\kappa}+1)^2+1$.
We present an example here.
\begin{example}\label{ex_g}
There are two types of jobs: long and short. The length of long jobs is $\kappa$, while the length of short jobs is $1$. Let $p$ be a large integer, the number of long and short jobs are $p+1$ and $p\kappa$, respectively. The fist long job $J_0^l$ is released at time $0$, and its type is
$\theta_0^l=(0, \kappa, 1, \kappa, 1)$. For $p-2 \geq i \geq 1$, job $J_i^l$ has type $\theta_i^l=(i(\kappa-\epsilon), (i+1)\kappa, 1, \kappa, f^i(1))$. Job $J_p^l$ has type $\theta_p^l=(p(\kappa-\epsilon), (p+1)\kappa, 1, \kappa, f^{p-1}(1)\cdot f(1-\epsilon)-\delta)$. Here, $\epsilon$ and $\delta$ are small constants satisfying $p\epsilon \ll 1$ and $\delta \ll \epsilon$. We also have long job $J_{p-1}^l$, whose type is $\theta_{p-1}^l=((p-1)(\kappa-\epsilon), (p+2)\kappa, 1, \kappa, f^{p-1}(1)+\delta)$. In the meanwhile, we have short jobs as follows. For $j=0,1,\ldots,p\kappa-1$, we denote $J_j^s$ as the $(j+1)$th short job, whose type is $\theta_j^s=(j, j+1, 1, 1, f^{\lfloor j/\kappa \rfloor}(1)\cdot f(j/\kappa - \lfloor j/\kappa \rfloor)-\delta/\kappa)$. for $j=0,1,\ldots,p\kappa-1$.
\end{example}

It can be verified that only one job $J_{p-1}^l$ can be completed in our mechanism, with a social welfare $\sim f^{p-1}(1)$. While in an optimal allocation, all the short jobs will be completed, and after that, $J_p^l$ and $J_{p-1}^{l}$ will be completed successively, with a social welfare
$\sim (f(0)+f(\frac{1}{\kappa})+,\ldots,+f(\frac{\kappa-1}{\kappa}))\cdot(1+\ldots+f^{p-1}(1))+f^p(1)+f^{p-1}(1) \geq \kappa\cdot (1+f(1)+,\ldots,+f^{p-1}(1))+f^p(1)+f^{p-1}(1)$. Therefore, the competitive ratio of our mechanism is at least $\kappa\cdot (f^{-(p-1)}(1)+\ldots+1)+f(1)+1$, which tends to $\frac{\kappa}{1-f^{-1}(1)}+f(1)+1$, when $p \rightarrow \infty$.
We use $\alpha$ to denote this competitive ratio, i.e. $\alpha=\frac{\kappa}{1-f^{-1}(1)}+f(1)+1$. Regarding $\alpha$ as a function of $f(1)$, we have $\alpha= \frac{\kappa\cdot f(1)}{f(1)-1}+f(1)+1=\kappa+\frac{\kappa}{f(1)-1}+(f(1)-1)+2$. Because $f(1)\geq 1$, it is clear that $\alpha\geq \kappa+2\sqrt{\kappa}+2 = (\sqrt{\kappa}+1)^2+1$, and equality holds if and only if $f(1)=\sqrt{\kappa}+1$. Therefore, the following example holds.

\begin{theorem}\label{general}
When $C=1$ and $f$ is a non-decreasing function, the competitive ratio of $\Gamma_G$ is no less than $(\sqrt{\kappa}+1)^2+1$.
\end{theorem}

Actually, when $C\geq h\cdot n_{\max}$, where $h\geq 2$ is an integer, the competitive ratio of $\Gamma_G$ with a non-decreasing priority function $f$ have a lower bound $\frac{h}{h-1}\cdot (\sqrt{\kappa}+1)^2+1$.

\begin{theorem}\label{general}
Assume $C\geq h\cdot n_{\max}$, where $h\geq 2$ is an integer, with any general non-decreasing priority function $f$, the competitive ratio of $\Gamma_G$ is no less than $\frac{h}{h-1}\cdot(\sqrt{\kappa}+1)^2+1$.
\end{theorem}

\subsection{Linear Priority Functions}
Since the linear priority function has been widely used in online scheduling problems \cite{porter2004mechanism}, in this subsection, we study how the proposed mechanism performs with a linear priority function:
$$f(\delta)=1+a\delta,$$
where $a$ is a non-negative parameter.

\begin{theorem}\label{linear}
When $C=1$, the competitive ratio of $\Gamma_G$ with a linear priority function is lower bounded by $(\sqrt{2\kappa(\kappa+1)}+\frac{3}{2}(\kappa+1))$.
\end{theorem}

\begin{proof}
To find the lower bound, we still use Example~\ref{ex_g}. Now the priority is linear, so we can represent $f(\delta_i)$ as $(1+a\delta_i)$, here $a$ is a constant non-negative parameter. So in Example~\ref{ex_g}, social welfare obtained by the optimal solution is $\sim (f(0)+f(\frac{1}{\kappa})+,\ldots,+f(\frac{\kappa-1}{\kappa}))\cdot(1+\ldots+f^{p-1}(1))+f^p(1)+f^{p-1}(1)= ((1+\frac{0}{\kappa}\cdot a)+(1+\frac{1}{\kappa}\cdot a)+,\ldots,+(1+\frac{k-1}{k}\cdot a))\cdot(1+(1+a)+,\ldots,+(1+a)^{p-1})+(1+a)^p+(1+a)^{p-1}=(\kappa+\frac{\kappa-1}{2}\cdot a)\cdot (1+(1+a)+,\ldots,+(1+a)^{p-1})+(1+a)^p+(1+a)^{p-1}$. However, social welfare obtained by our mechanism is $\sim f^{p-1}(1)=(1+a)^{p-1}$. Therefore, the competitive ratio of our mechanism is at least $(\kappa+\frac{\kappa-1}{2}\cdot a)\cdot ((1+a)^{-(p-1)}+\ldots+1)+(1+a)+1$, which tends to $(\kappa+\frac{\kappa-1}{2}\cdot a)\cdot\frac{1}{1-(1+a)^{-1}}+(1+a)+1$, when $p \rightarrow \infty$. We use $\beta$ to denote this competitive ratio, i.e. $\beta=(\kappa+\frac{\kappa-1}{2}\cdot a)\cdot\frac{1}{1-(1+a)^{-1}}+(1+a)+1=\frac{\kappa}{a}+(\frac{\kappa+1}{2})a+\frac{3}{2}(\kappa+1)$. Because $a\geq 0$, we have $\beta\geq \sqrt{2\kappa(\kappa+1)}+\frac{3}{2}(\kappa+1)$, and equality holds if and only if $a=\sqrt{\frac{2\kappa}{\kappa+1}}$.
\end{proof}

It is easy to extend this result to the multi-instance case. The proof is very similar and is omitted. We state the theorem here.

\begin{theorem}\label{linear}
Assume $C\geq h\cdot n_{\max}$, where $h\geq 2$ is an integer. The competitive ratio of $\Gamma_G$ with a linear priority function is lower bounded by $\frac{h}{h-1}\cdot (\sqrt{2\kappa(\kappa+1)}+\frac{3}{2}\kappa+\frac{1}{2})+1$.
\end{theorem}

\subsection{Discussions}
In this subsection, we make some discussions about the results we obtained so far.

First, we compare the performance of $\Gamma_G$ with different priority functions.
With some derivations, one can verify
$$(\sqrt{2\kappa(\kappa+1)}+\frac{3}{2}(\kappa+1))>(\kappa+1)(1+\frac{1}{\kappa})^{\kappa}+1$$ and
$$\frac{h}{h-1}\cdot (\sqrt{2\kappa(\kappa+1)}+\frac{3}{2}\kappa+\frac{1}{2})+1>\frac{h}{h-1}\cdot \frac{\chi}{1-\chi^{-1/\kappa}}+1.$$  This observation suggests that, the greedy mechanism with an exponential priority function performs better than that with a linear priority function.

Second, let us look at a simple model, in which each job $i$ has unit length ($l_i=1$) and needs only one instance ($n_i=1$) to process it. And there are only one instance in the cloud, i.e. $C=n_{\max}=1$. This is a special case in our general model, so all our theorems apply. For this simple case, we have the following corollary.

\begin{corollary}\label{yes5}
For the simple case, $\Gamma_{G}$  with an exponential priority function can achieve a competitive ratio of $5$.
\end{corollary}

 Corollary~\ref{yes5} can be directly derived from Theorem~\ref{main2}. Note that in the simple case, $\kappa=1$, and we choose $\chi=2$ to have a $5$-competitive mechanism. The results accord with Theorem $8$ in \cite{hajiaghayi2005online}.

\section{A Dynamic Program Based Mechanism}\label{dpm}
The mechanism studied in previous section takes a simple greedy approach to select a set of valuable jobs from all the feasible jobs at each critical time point. It is easy to see that, given the virtual value $v_i'$ and the demanded instances $n_i$ of each feasible job and the capacity $C$ of the cloud, a better approach is to use the dynamic program designed for the knapsack problem to select a set of most valuable jobs. In this section, we design such a mechanism, denoted as $\Gamma_{D}$.

As shown in Algorithm 2, the allocation rule of $\Gamma_D$ is based a dynamic program. Its payment rule is the same as the previous greedy mechanism: charge each agent according to his/her critical value. In the remaining part of this section, we prove a tight competitive bound (Theorem \ref{dp}) for $\Gamma_D$, the first step of which is to lower bound the competitive ratio of the mechanism (Lemma \ref{nmax}) .

\begin{algorithm}[H]
\For{all critical time point $t$ in the ascending order}{
$J_F\leftarrow J_F(t)$\;
\eIf{$J_F\neq\emptyset$}{
For each $i\in J_F$, update the virtue value $v'_i=v_i\cdot\chi^{\delta_i}$\;
Using the dynamic programming algorithm, find the most valuable (in terms of $v'_i$) set of jobs, denoted as $S_t$\;
Run $S_t$\;
}{
Output $\emptyset$;
}
}
\caption{The allocation rule of $\Gamma_D$}
\end{algorithm}

\begin{lemma}\label{nmax}
The competitive ratio of  $\Gamma_D$ is at least $n_{\max}\cdot \frac{\chi}{1-\chi^{-1/\kappa}}+1$.
\end{lemma}

\begin{proof}
We prove this by an example. For the convenience of analysis, we assume $\kappa$ is an integer and $l_{\min}=1$. In our example, $C=h\cdot n_{\max}$, and there are two types of jobs: long and short. The length of long jobs is $\kappa$, while the length of short jobs is $1$.

The long jobs are released by groups. Let $p$ be a large integer, $\epsilon$ and there are $p$ groups of long jobs.

The first ``long-job'' group (denoted as $J_0^l$) consists of $h$ long jobs whose types are $\theta_0^l=(0, \kappa, n_{\max}, \kappa, 1)$.

The $(i+1)$-th ``long-job'' group (denoted as $J_i^l$) consists of $h$ long jobs whose types are $\theta_i^l=(i(\kappa-\epsilon), (i+1)\kappa, n_{\max}, \kappa, \chi^i)$, here $p-2 \geq i \geq 1$.

The $p$-th ``long-job'' group (denoted as $J_{p-1}^l$) consists of $h$ long jobs whose type are $\theta_i^l=((p-1)(\kappa-\epsilon), (p+2)\kappa, n_{\max}, \kappa, \chi^i)$.

Here, $\epsilon$ are small constants satisfying $p\epsilon \ll 1$.

The short jobs are released by queues, and there are $h\cdot n_{\max}$ queues of short jobs. In each queue, there are $p\kappa$ short jobs released one by one.

In the $k$-th ``short-job'' queue (denoted as $J_j^s$), we have such jobs:
the $(j+1)$-th short job in the $k$-th ``short-job'' queue is $\theta_{kj}^s=(1-p\epsilon-k\delta+j, 1-p\epsilon-k\delta+j+1, 1, 1, \chi^{\frac{1-p\epsilon-k\delta+j}{\kappa}}-\frac{\delta}{p\kappa})$, for $j=0,1,\ldots,p\kappa-1$.

Here, $h\cdot n_{\max} \cdot \delta \ll \epsilon$.

It can be verified that only group $J_{p-1}^l$ can be completed in our mechanism, with a social welfare $\sim h\cdot \chi^{p-1}$. While in optimal solution, all the short jobs will be completed, and after that, group $J_{p-1}^l$ will be completed successively, with a social welfare
$\sim h\cdot n_{\max}(1+\chi^{1/\kappa}+,\ldots,+\chi^{p-1/\kappa}+\chi^{p})+ h\cdot \chi^{p-1}$. Therefore, the competitive ratio of our mechanism is at least $\sim n_{\max}\cdot (\chi^{-p+1}+\chi^{-p+1+1/\kappa}+,\ldots,+\chi^{1-1/\kappa}+\chi)+ 1$, which tends to $n_{\max}\cdot \frac{\chi}{1-\chi^{-1/\kappa}}+1$, when $p \rightarrow \infty$.
\end{proof}

\begin{theorem}\label{dp}
The mechanism $\Gamma_D$ has a competitive ratio of $n_{\max}\cdot\frac{\chi}{1-\chi^{-\frac{1}{\kappa}}}+1$.
\end{theorem}
\begin{proof}
From Lemma~\ref{nmax}, we know that $n_{\max}\cdot\frac{\chi}{1-\chi^{-\frac{1}{\kappa}}}+1$ is a competitive lower bound of $\Gamma_D$,
we now prove that it is also an upper bound.
We will still charge the values of winning jobs in an optimal solution
OPT to winning jobs in our algorithm.

For any winning agent $i$ in OPT, if she is also a winner in our algorithm,
then her value is charged to herself.
Otherwise, consider the time $t$ at which $i$ is allocated the instances in OPT. At this time,
our algorithm allocates at least $C-n_i+1$ instances to other jobs,
since $i$ is not allocated. We use $J_i$ to denote the set of jobs
that are active at time $t$ in our algorithm. We give the following claim.
\paragraph{Claim 1}
Jobs in $J_i$ have a total value (in terms of $v'_j$) of at least $\frac{C}{n_in_{\max}}v_i$.

We first prove this claim. We use $V_{i}$ to denote this total value. It is clear that when $n_i\geq\sqrt{C}$ the conclusion holds, since
$V_{i}\geq v_i$ and $\frac{C}{n_in_{\max}}\leq1$. So in the following, we assume $n_i<\sqrt{C}$.
When $n_i=1$, then our algorithm allocates all the instances to $J_i$, and there are at least $\lceil \frac{C}{n_{\max}}\rceil$ jobs in $J_i$,
since otherwise $i$ will be allocated. Besides, each job $j\in J_i$ has a value no less that $v_i$ (in terms of $v'_j$). Otherwise $j$ will
be replaced by $i$. Therefore $V_{i}\geq \lceil \frac{C}{n_{\max}}\rceil v_i\geq \frac{C}{n_{\max}}v_i=\frac{C}{n_in_{\max}}v_i$.

When $n_i\geq2$, we assume that there are $\eta$ jobs in $J_i$ whose size is no smaller than $n_i$. Each of these job has a value greater than $v_i$.
There are at least $C-n_i+1-\eta\cdot n_{\max}$ instances allocated to jobs whose size is smaller than $n_i$. Since all the $n_j$s are integer,
there are at least $\lceil\frac{C-n_i+1-\eta\cdot n_{\max}}{n_i-1}\rceil$ small jobs.
We can combine these small jobs to at least $\lceil\frac{C-n_i+1-\eta\cdot n_{\max}}{2(n_i-1)}\rceil$ large jobs, each has a size larger than $n_i$
We make this combination in the following way: Giving each job a label which from $1$ to $\lceil\frac{C-n_i+1-\eta\cdot n_{\max}}{n_i-1}\rceil$.
Starting from the first job, we use as few jobs as possible to combine them to a large job which has size no less than $n_i$. Each time
we have a waste of at most $n_i-2$ size, since every small job has a size no more than $n_i-1$.
Therefore we get at least $\lceil\frac{C-n_i+1-\eta\cdot n_{\max}}{2(n_i-1)}\rceil$ large jobs, and each has a value larger than $v_i$,
which implies that $V_i\geq\lceil\frac{C-n_i+1-\eta\cdot n_{\max}}{2(n_i-1)}+\eta\rceil v_i$. If $2(n_i-1)\leq n_{\max}$,
then $V_i\geq\lceil\frac{C-n_i+1-\eta\cdot n_{\max}}{2(n_i-1)}+\eta\rceil v_i\geq \lceil\frac{C-\frac{1}{2}n_{\max}}{n_{\max}}\rceil v_i\geq\lceil\frac{C}{2n_{\max}}\rceil v_i\geq\lceil\frac{C}{n_in_{\max}}\rceil v_i$. Otherwise, if $2(n_i-1)>n_{\max}$, then $V_i\geq \lceil\frac{C-n_i+1}{2(n_i-1)}\rceil v_i\geq\lceil\frac{C}{n_in_{\max}}\rceil v_i$, since $(C-n_i+1)n_in_{\max}\geq n_in_{\max}C-2C(n_i-1)\geq 2C(n_i-1)$, where the first inequality holds by $\frac{n_{\max}}{2}+1<n_i<\sqrt{C})$, and
the second inequality holds by $n_{\max}\geq n_i$. This complete the proof of Claim 1.~\qed\\

We continue the proof of Theorem~\ref{dp}.
We first temporarily charge the value of $i$ to all the jobs in $J_{i}$ in proportion with their values, and each job $j\in J_{i}$ is temporarily charged $\frac{v'_j}{V_i}v_i$ ($\leq \frac{n_in_{\max}}{C}v'_j$ by Claim 1).
A job $j\in J_{i}$ might be interrupted in our algorithm. If she is not interrupted, then it is finally charged $\frac{v'_j}{V_i}v_i$.
If she is interrupted at time $t'$, we use $J_{int}$ to denote all these jobs that are interrupted at $t'$.
We then pass all the temporary charge of $j\in J_{int}$ to jobs that are newly allocated at $t'$ also in proportion with their values, and
other jobs in $J_{i}\backslash J_{int}$ keep their temporary charge.
Note that by the dynamic programming algorithm the total value of new allocated jobs is no less than
that of the interrupted jobs.
Therefore, after the interruption, each job $j$ has a temporary charge of at most $\frac{n_in_{\max}}{C}v'_j$.
We continue this chain until all the temporary charge is finally charged.
We now calculate the maximum total value charged to a agent $j$ with
value $v_j$ who wins at time $t$ in our algorithm.

If job $j$ is completed in $OPT$, there is a charge of $v_j$. Divide all jobs in $OPT$ whose value
is charged to $j$ to different groups according to their start time in $OPT$ by the following rule:

Consider a job $i$ in $OPT$ whose value is charged to $j$. Let $t' = t- \sigma_i$ be the time at which job $i$ receives an allocation in $OPT$, then we say $i$ is in group $\sigma_i$.
It is clear from the mechanism that $\sigma_i > -l_j$.

When $\sigma_i\leq 0$, it is easy to see that, the value charged to $j$ by job $i$ is at most $\chi^{-\sigma_i/l_j}\cdot\frac{n_in_{\max}}{C}v_j$. Thus, the total charge from group $\sigma_i$ is at most $\sum_{i}\frac{n_in_{\max}}{C}v_j\chi^{-\sigma_i/l_j}\leq n_{\max}\chi^{-\sigma_i/l_j}v_j$, since if we use $Number_i$ to denote the number of jobs whose size is
$n_i$ in group $\sigma_i$, then $\sum^{n_{\max}}_{n_i=1}n_i\cdot Number_i\leq C$, which implies our inequality.

When $\sigma_i>0$, similar to the proof of Theorem~\ref{main1} we can that the value charged to $j$ by job $i$ is at most $\chi^{-\sigma_i/\kappa}\cdot\frac{n_in_{\max}}{C}v_j$, and the total charge from group $\sigma_i$ is at most $\sum_{i}\frac{n_in_{\max}}{C}v_j\chi^{-\sigma_i/\kappa}\leq n_{\max}\chi^{-\sigma_i/\kappa}v_j$
Also, the value of $\sigma_i$ for any two such groups must be apart by at least $l_{min}=1$, so $\sigma_i=-l_j+i\cdot 1$, for $i=0,1,2,\ldots,+\infty$. Therefore, the total charge to $j$ is at most
\begin{align}
\nonumber
&v_j+n_{\max}\cdot \sum_{i: \sigma_i=0}^\infty \chi^{-\frac{\sigma_i}{\kappa}} v_j + n_{\max}\cdot \sum_{i:\sigma_i=-\kappa}^{i:\sigma_i=-1} \chi^{-\frac{\sigma_i}{l_j}} v_j\\
\nonumber
\leq& v_j+n_{\max}\cdot \sum_{i: \sigma_i\geq0}^\infty \chi^{-\frac{\sigma_i}{\kappa}} v_j + n_{\max}\cdot \sum_{i: \sigma_i<0} \chi^{-\frac{\sigma_i}{\kappa}} v_j\\
\nonumber
=&(n_{\max}\cdot \frac{\chi}{1-\chi^{-\frac{1}{\kappa}}}+1)v_j.
\end{align}

This shows that our algorithm is $(n_{\max}\cdot \frac{\chi}{1-\chi^{-\frac{1}{\kappa}}}+1)$-competitive.~\qed
\end{proof}

\begin{remark}
As shown in Lemma \ref{lem3}, the competitive ratio of the greedy mechanism $\Gamma_G$ with an exponential priority function can be arbitrarily bad if there exists some agent with $n_i\rightarrow C$. That is, its worst case performance is not guaranteed. In contrast, the competitive ratio of $\Gamma_D$ is always upper bounded; in particular, it is much better than $\Gamma_G$ with an exponential priority function when the demand $n_i$ of some agent is very large and close to $C$. This is consistent with our intuition, since $\Gamma_D$ leverages a more complex approach to select the most valuable jobs at each critical time point.
\end{remark}

\section{Conclusion and future work}
In this paper, we have studied the problem of online mechanism design for resource allocation and pricing in cloud computing (RAPCC).  We have shown that the optimal online allocation for RAPCC is NP-hard. Then two kinds of DSIC online mechanisms have been designed: the first one, which is based on a greedy allocation rule, is very fast and has a tight competitive bound; and the second one, which is based on a dynamic program for allocation, is relatively computationally expensive but with a better competitive bound when the maximum demand $n_{\max}$ of agents is close to the supply (i.e., the capacity $C$) of the cloud provider.

There are many aspects to explore about online mechanisms for RAPCC in the future. First, in this paper we have focused on deterministic online mechanisms. We would like to explore randomized mechanisms in the future. Second, we have ignored agents' temporal strategies and assumed that they will not misreport their arrival times and deadlines. We will take the temporal strategies into consideration. Third, in cloud computing, both the supply (e.g., different kinds of virtual machines) of the provider and the demands of cloud customers can be heterogeneous, which is beyond our current setting. We will study online mechanisms for such complex settings.

\bibliographystyle{acmsmall}
\bibliography{cc}

\end{document}